%% file: main.tex
\documentclass[10pt,conference]{IEEEtran}

\usepackage{cite}
\usepackage{amsmath,amssymb,amsfonts,amsthm}
\usepackage{mathtools}
\usepackage[ruled,lined,noend]{algorithm2e}
\usepackage{graphicx}
\usepackage{textcomp}
\usepackage[dvipsnames]{xcolor}
\usepackage{url}
\usepackage{subfig}

\definecolor{orange}{RGB}{255,140,0}
\definecolor{purple}{RGB}{138,43,226}

\theoremstyle{definition}
\newtheorem{theorem}{Theorem}

\newtheorem{definition}{Definition}

\renewcommand{\baselinestretch}{0.98}
\colorlet{c1}{red} 
\colorlet{c2}{black} 
\colorlet{c3}{black} 

\newcommand{\ProbName}{TSCC}
\newcommand{\AlgName}{ETSCAA}
\newcommand{\AlgStepOne}{VPTS}
\newcommand{\AlgStepTwo}{TQA}
\newcommand{\AlgStepThree}{CTQC}
\newcommand{\IndiOne}{CI}
\newcommand{\IndiTwo}{VLI}
\newcommand{\IndiThree}{SMI}

\begin{document}

\title{Cybersickness-aware Tile-based Adaptive 360\textdegree\ Video Streaming}

\author{
Chiao-Wen Lin$^{\dag}$, Chih-Hang Wang$^{\sharp}$, De-Nian Yang$^{\sharp}$, and Wanjiun Liao$^{\dag}$\\
$^{\dag}$Department of Electrical Engineering, National Taiwan University, Taiwan\\
$^{\sharp}$Institute of Information Science, Academia Sinica, Taiwan\\
E-mail: d09921012@ntu.edu.tw, \{superwch7805, dnyang\}@iis.sinica.edu.tw, and wjliao@ntu.edu.tw\\
}


\maketitle

\input{1-Abstract}


\IEEEpeerreviewmaketitle

\section{Introduction} \label{intro}
\input{2-Introduction}

\input{3-System}

\input{4-Problem}

\input{5-Algorithm}
\vspace{-4mm}
\input{6-Simulation}
\input{8-Conclusion}
\vspace{-4mm}

\ifCLASSOPTIONcompsoc
  \section*{Acknowledgments}
\else
  \section*{Acknowledgment}
\fi

We thank to National Center for High-performance Computing (NCHC) of National Applied Research Laboratories (NARLabs) in Taiwan for providing computational and storage resources.

\bibliographystyle{IEEEtran}
\bibliography{9-Reference}

\end{document}

%% file: 1-Abstract.tex
\begin{abstract}
In contrast to traditional videos, the imaging in virtual reality (VR) is 360\textdegree, and it consumes larger bandwidth to transmit video contents.
To reduce bandwidth consumption, tile-based streaming has been proposed to deliver the focused part of the video, instead of the whole one.
On the other hand, the techniques to alleviate cybersickness, which is akin to motion sickness and happens when using digital displays, have not been jointly explored with the tile selection in VR.
In this paper, we investigate Tile Selection with Cybersickness Control (\ProbName) in an adaptive 360\textdegree\ video streaming system with cybersickness alleviation. We propose an $\boldsymbol{m}$-competitive online algorithm
with Cybersickness Indicator (\IndiOne) and Video Loss Indicator (\IndiTwo) to evaluate instant cybersickness and the total loss of video quality. 
Moreover, the algorithm exploits Sickness Migration Indicator (\IndiThree) to evaluate the cybersickness accumulated over time and the increase of optical flow to improve the tile quality assignment.
Simulations with a real network dataset show that our algorithm outperforms the baselines regarding video quality and cybersickness accumulation.
\end{abstract}

%% file: 2-Introduction.tex
Expedited by the Covid-19 pandemic in 2020, much of the world has been transformed into collaboratively operating in a virtual environment \cite{Forbes1}. It markedly accelerates the growth of virtual reality (VR) by necessity.
For example, Facebook has launched Facebook Horizon to allow people to engage in a variety of activities like riding roller coasters with different themes. Users can play fast-paced shooting and car racing games with SteamVR.\footnote{Facebook Horizon: https://www.oculus.com/facebook-horizon/, SteamVR: https://store.steampowered.com/app/250820/SteamVR/}
Equipped with VR head-mounted displays (HMDs), such as Oculus Rift, HTC VIVE, and Google Cardboard,\footnote{Oculus Rift: https://www.oculus.com, HTC VIVE: https://www.vive.com, Google Cardboard: https://arvr.google.com/cardboard} a user can arbitrarily choose a view of the 360\textdegree\ video displaying the virtual world by moving her head dynamically. However, the bandwidth demand for streaming a 360\textdegree\ video is an order of magnitude larger than traditional static videos (e.g., movies) \cite{360vid-summary}. 


To reduce bandwidth consumption, 
tile-based streaming is proposed to stream a part of a video (i.e., video tiles) displayed \cite{360vid-ProbDASH,360vid-Flare,360vid-BAS,360vid-TwoTier,360vid-EPASS360}. 
It first cuts a 360\textdegree\ video into multiple tiles spatially and further encodes these tiles into different quality for adaptive streaming. 
Xie \emph{et al.} \cite{360vid-ProbDASH} chose tiles and decided their quality based on viewport distortion, quality variance, and buffer control. Qian \emph{et al.} \cite{360vid-Flare} selected tiles according to the tile ranking and classing to maximize the utility function under the consideration of inter-chunk switch, intra-chunk switch, and stall duration. Xiao \emph{et al.} \cite{360vid-BAS} combined tiles into macro-streaming units and chose quality levels to reduce bandwidth waste. Sun \emph{et al.} \cite{360vid-TwoTier} dynamically adjusted the rate allocation and buffer length based on the network bandwidth statistics and viewport prediction. 
However, these works only focused on bandwidth reduction but did not explore cybersickness during the tile selection and buffer management in VR.

\emph{Cybersickness} is akin to motion sickness, but it happens when using VR HMDs, and the user may have symptoms of headaches, nausea, and so on \cite{Sickness-TempAspect}.
Since these feelings degrade the user's experience with VR, it is vital to examine cybersickness during the tile selection in VR \cite{360vid-summary}. Several crucial factors lead to cybersickness.
1) \emph{User head movement}. Faster head rotations incur more cybersickness \cite{Sickness-DesignGuideline}. When the user rotates head faster, the Motion-to-photon (MTP) delay increases due to the rapid change of frames. MTP delay is the time difference between the user's head movement and the VR frame display (reflecting the user's movement) \cite{360vid-summary}. 
2) \emph{Magnitude of optical flow}. 
Optical flow is the velocity distribution of apparent movements and changes in brightness patterns between two similar images. The fast motions and rapid changes of brightness and color of objects induce a large amount of optical flow. A VR scene with a large magnitude of optical flow triggers cybersickness because it contains massive visual information, amplifying the mismatch between the real world and the VR scene \cite{Sickness-OpticalFlow}. 
3) \emph{Task duration}. 
Longer task duration to use VR induces more cybersickness since cybersickness accumulates over time \cite{Sickness-DesignGuideline,Sickness-TempAspect,Sickness-UserAdaptation}.

To alleviate cybersickness, field of view (FoV) shrinking \cite{Sickness-FoV-1,Sickness-FoV-2} and depth of field (DoF) simulation \cite{Sickness-DoF,Sickness-DesignGuideline} have been demonstrated as two promising approaches. FoV shrinking slightly and temporarily reduces the size of the viewport by masking the outer part of the original video \cite{Sickness-FoV-1,Sickness-FoV-2}. Smaller FoV induces less cybersickness with the side effect of lowering the sense of immersion, since it reduces the amount of visual information received.
On the other hand, DoF simulation further blurs the out-of-focus parts of a video to decrease the magnitude of optical flow when the view changes rapidly \cite{Sickness-DoF,Sickness-DesignGuideline}.
However, these works only verified the effectiveness of FoV shrinking and DoF simulation but did not explore the trade-off between cybersickness and video quality in mobile VR. 



In this paper, we explore Tile Selection with Cybersickness Control (\ProbName) in an adaptive 360\textdegree\ video streaming system. 
Unlike previous research of 360\textdegree\ video streaming ignoring the impact of cybersickness \cite{360vid-BAS,360vid-EPASS360,360vid-ProbDASH,360vid-Flare,360vid-TwoTier}, \ProbName\ minimizes the video quality loss and cybersickness accumulation by jointly fetching tiles, shrinking the user's viewport, and determining the activation of DoF simulation.
\ProbName\ is challenging due to the following new research issues.
1) \emph{Trade-off between cybersickness and video quality}.
A high-quality video contains a large magnitude of optical flow, which induces severe cybersickness \cite{Sickness-OpticalFlow}. Hence, when users feel dizzy, it is necessary to constrain cybersickness by slightly limiting the viewport size with FoV shrinking and blurring out-of-focus tiles with DoF simulation to reduce the magnitude of optical flow. The joint configuration of the quality of video tiles (at different regions), FoV shrinking, and DoF simulation is crucial to VR. 
2) \emph{Cybersickness accumulation}.
As the user keeps enjoying VR without pauses, cybersickness will aggravate over time since the accumulation is mostly faster than the speed of human adaptation to VR \cite{Sickness-TempAspect,Sickness-UserAdaptation}. However, activating FoV shrinking and DoF simulation for a long period will dramatically decrease the user's experience with the video. Therefore, it is crucial to decide the activation and stop time of FoV shrinking and DoF simulation carefully. 
3) \emph{Packet queue control}.
{\color{c3}High-quality tiles consume much bandwidth and lead to longer transmission time, which may cause video stalls (i.e., the packet queue underflows). Most research \cite{360vid-ProbDASH,360vid-TwoTier} controls the packet queue based on only the network bandwidth and aims to maximize video quality. However, it is necessary to examine user behaviors (e.g., head rotation) when controlling the packet queue in VR since they may induce higher cybersickness (which needs to be compensated by lowering video quality).
Therefore, the packet queue management becomes much more complicated to avoid the overflow and underflow of the queue.
}

To address the above issues, we design an $m$-competitive\footnote{$m$ is the ratio of number of tiles chosen by \AlgName\ to the optimal policy.} online algorithm, named Effective Tile Selection and Cybersickness Alleviation Algorithm (\AlgName).
\AlgName\ introduces Cybersickness Indicator (CI) and Video Loss Indicator (VLI) to evaluate the instant cybersickness and the total loss of video quality by carefully assessing the impacts of FoV shrinking and DoF simulation at different tiles. {\color{c3}Equipped with CI and VLI, \AlgName\ assigns tile quality according to the induced costs (including expected viewport distortion and cybersickness accumulation) and the available bandwidth to control the packet queue growth. Moreover, \AlgName\ introduces Sickness Migration Indicator (SMI) to evaluate the cybersickness accumulated over time and the increase of optical flow. \AlgName\ adjusts tile quality to jointly alleviate cybersickness and optimize video quality based on SMI.}


The rest of this paper is organized as follows. 
Section \ref{sys} introduces the 360\textdegree\ video streaming system and Section \ref{pro} formulates \ProbName.
Section \ref{alg} presents the online algorithm \AlgName\ and proves its competitive ratio. Section \ref{sim} shows the simulation results, and Section \ref{con} concludes this paper.

%% file: 3-System.tex
\section{System Model} \label{sys}

Following \cite{360vid-summary,360vid-BAS,360vid-EPASS360}, a 360\textdegree\ video is first segmented into small chunks for storage. 
Next, each chunk is cropped into tiles spatially, and each tile is encoded into multiple bitrates by the video compression, such as H.264 or H.265, and stored in a mobile proxy server \cite{360vid-summary}. 
When a user request arrives, the server will combine the fetched tiles in the cache to render the required video frames, where the cached tiles were chosen from the server based on the user viewport prediction \cite{360vid-summary}. 

To alleviate cybersickness, the system is equipped with FoV shrinking \cite{Sickness-FoV-1,Sickness-FoV-2} and DoF simulation \cite{Sickness-DoF,Sickness-DesignGuideline}.
After combining the tiles to form a complete video, the requested view is extracted from the video.
The bokeh blur filter \cite{Sickness-DoF} is then applied to the viewport to create the DoF simulation, followed by slightly masking the upper and lower side of the viewport to shrink FoV. 
Afterward, the video quality degradation caused by FoV shrinking and DoF simulation can be evaluated according to structural similarity (SSIM) index and bitrate \cite{SSIM-1,SSIM-2}. 

The system includes a packet queue and a sickness queue maintained on the user side. The packet queue stores the video rendered from the fetched tiles for smooth playback\footnote{To avoid cybersickness caused by MTP delay, it is necessary to maintain the occupancy of packet queue (detailed later) \cite{360vid-ProbDASH,360vid-summary}.} \cite{360vid-ProbDASH,360vid-Flare,360vid-BAS,360vid-TwoTier}, and the queue length increases and decreases when new frames are downloaded and the video is played, respectively.
The sickness queue (derived according to the Time-varying Cybersickness Model \cite{Sickness-TempAspect,Sickness-UserAdaptation,Sickness-OpticalFlow,Sickness-DoF,Sickness-FoV-1,Sickness-FoV-2,Sickness-DesignGuideline} and detailed in Section \ref{pro}) records cybersickness accumulation when the video is played. {\color{c2}According to \cite{Sickness-DesignGuideline,Sickness-TempAspect,Sickness-UserAdaptation}, we measure cybersickness severity by the Simulator Sickness Questionnaire (SSQ) metric, which can be calculated by the magnitude of optical flow and user head rotation speed, monitored by eye and neck movement tracking schemes in VR HMDs \cite{Sickness-DesignGuideline}.
The queue length increases when the magnitude of optical flow and user head rotation speed increase, whereas it constantly decreases due to the user adaptation to VR.
}
{\color{c3}The packet queue and the sickness queue affect each other when the available bandwidth, user behavior, or the amount of video optical flow changes.}

%% file: 4-Problem.tex
\section{Problem Formulation} \label{pro}
Due to the space constraint, the notation table is provided in \cite{tech}.
We formulate \ProbName\ for the tile selection and cybersickness control in a 360\textdegree\ video streaming system. 
Specifically, the system includes a 360\textdegree\ video divided into multiple chunks with the duration $T$ \cite{360vid-EPASS360}, and each chunk is spatially cropped into $N$ tiles \cite{360vid-EPASS360}, where each tile is encoded into $L$ quality levels \cite{360vid-EPASS360}.
Let $(i,j)$ represent tile $i$ with quality $j$. Each tile $(i,j)$ induces bandwidth consumption $b_{i,j}$, distortion $d_{i,j}$, and the magnitude of optical flow $f_{i,j}$, which are specified in the video metadata \cite{360vid-ProbDASH}.
We describe the system status at time slot $t$ as $S_t=\{R_t,D_t,Q_{t}^P,Q_{t}^S\}$, where $R_t$, $D_t$, $Q_{t}^P$, and $Q_{t}^S$ are the status of user head rotation, expected viewport distortion, packet queue occupancy, and sickness queue occupancy at time slot $t$, respectively, introduced as follows.

1) The status of \textbf{user head rotation} $R_t=(\omega^y_t,\omega^p_t)$ includes the angular speed of yaw rotation $\omega^y_t$ and the angular speed of pitch rotation $\omega^p_t$ \cite{360vid-ProbDASH,360vid-Flare,360vid-BAS,360vid-TwoTier}. 
Let $p_{i,t}$ be the viewing probability of tile $i$ at time slot $t$, which can be derived by the viewport prediction with the stochastic model and historical data of user head rotation \cite{360vid-ProbDASH,360vid-Flare}.

2) The \textbf{expected viewport distortion} \cite{360vid-ProbDASH} $D_t=\frac{\Sigma_{i=1}^N \Sigma_{j=1}^L x_{i,j}p_{i,t}d_{i,j}}{\Sigma_{i=1}^N \Sigma_{j=1}^L x_{i,j}p_{i,t}}$ can be expressed as the average video quality distortion represented by the distortion of each tile $i$ weighted by its viewing probability $p_{i,t}$ at time slot $t$, where $x_{i,j}\in\{0,1\}$ is a decision variable indicating whether tile $(i,j)$ is fetched. The distortion of each tile is derived as the reciprocal of its SSIM index \cite{SSIM-1}, and a tile with higher quality has smaller distortion \cite{SSIM-1,SSIM-2}.

3) {\color{c3}Let $B_t$ be the bandwidth,\footnote{The bandwidth can be allocated by exploiting machine learning models, such as support vector machine and artificial neural network, that learn the relation between bandwidth allocation and quality of service \cite{alloc-1}.
} $C_p$ denotes the packet queue capacity,\footnote{Due to the short-term constraint of viewport prediction, tile-based streaming keeps a small packet queue capacity, whose unit is second, to store predicted viewports in the next few seconds \cite{360vid-ProbDASH,360vid-EPASS360}.} and $\gamma$ is the bandwidth consumption of the viewport.}
Since the packet queue length decreases when the download time of video chunk is longer than its duration $T$ (i.e., the video consumption rate is faster than download speed), following \cite{360vid-ProbDASH}, the \textbf{packet queue occupancy} $Q_{t}^P=Q_{t-1}^P+\frac{T}{C_p}(1-\frac{s_{fov}(1-k_{dof}y_{dof})\gamma}{B_t})$.
{\color{c3}$Q_{t-1}^P$ is the current occupancy, and $\frac{T}{C_p}(1-\frac{s_{fov}(1-k_{dof}y_{dof})\gamma}{B_t})$ represents the amount of occupancy change.\footnote{The decreasing occupancy estimated by the download time is subtracted from the increasing occupancy $\frac{T}{C_p}$ evaluated by the video chunk duration.}} $y_{dof}\in\{0,1\}$ is a decision variable indicating whether DoF simulation is activated, and $k_{dof}$ is the ratio of the bandwidth consumption reduced by DoF simulation. Following \cite{Sickness-FoV-1,Sickness-FoV-2}, $s_{fov}\in[0.7,1]$ is a decision variable representing the ratio of the shrunken viewport size to the original size when activating FoV shrinking.

4) When the user enjoys VR over time, cybersickness accumulates and the length of the sickness queue increases due to the head rotation \cite{Sickness-DesignGuideline,Sickness-FoV-1} and optical flow \cite{Sickness-OpticalFlow}. However, the user may gradually adapt to VR as time elapses. Let $\Omega$ be the consumption rate of the sickness queue (i.e., user adaptation to VR).\footnote{\color{c2}According to \cite{Sickness-UserDifference}, the values of $\Omega$ and $C_s$ for each user can be determined based on age, gender, and familiarity with VR contents. For example, seniors and people that are not familiar with VR contents tend to have a lower tolerance of cybersickness, and smaller $\Omega$ and $C_s$ are required.}
Following \cite{Sickness-TempAspect,Sickness-UserAdaptation}, the \textbf{sickness queue occupancy} $Q_{t}^S$ is derived according to i) $Q_{t-1}^S$, ii) increasing occupancy caused by head rotation \cite{Sickness-DesignGuideline,Sickness-FoV-1} and optical flow \cite{Sickness-OpticalFlow}, and iii) decreasing occupancy $\frac{\Omega}{C_s}$. That is, $Q_{t}^S=Q_{t-1}^S+\frac{(\frac{\sqrt{(\omega^y_t)^2+(\omega^p_t)^2}}{100\sqrt{2}}+\frac{\Sigma_{i=1}^N \Sigma_{j=1}^L x_{i,j}p_{i,t}f_{i,j}}{\Sigma_{i=1}^N \Sigma_{j=1}^L x_{i,j}p_{i,t}})s_{fov}(1-k_{dof}y_{dof})}{C_s}-\frac{\Omega}{C_s}$,
where $C_s$ is the sickness queue capacity,\footnotemark[\value{footnote}] and the optical flow of each tile is weighted by its viewing probability. $\frac{\sqrt{(\omega^y_t)^2+(\omega^p_t)^2}}{100\sqrt{2}}$ is the normalized head rotation speed \cite{360vid-Flare}. $s_{fov}(1-k_{dof}y_{dof})$ describes the percentage of cybersickness decreased due to FoV shrinking and DoF simulation.\footnote{The cybersickness reduction is inversely proportional to the viewport size \cite{Sickness-FoV-2} and the degree of video blurring \cite{Sickness-DoF}.}


\ProbName\ has the following constraints. 1) \emph{Fetching constraint}. The system fetches only one level for each tile at any instant to avoid tile redundancy, i.e.,  $\Sigma_{j=1}^L  x_{i,j}\leq1,\forall i$. 2) \emph{Packet queue length constraint.} To prevent video rebuffering and control MTP delay, the queue occupancy $Q_{t}^P$ must be no less than the targeted queue occupancy $\lambda$, i.e., $Q_{t}^P \geq \lambda$ \cite{360vid-ProbDASH,360vid-TwoTier}.
By properly reformulating the above equation (detailed in \cite{tech}), we can transform the packet queue length constraint into the \emph{bandwidth constraint} for fetching tiles, i.e., $\Sigma_{i=1}^N \Sigma_{j=1}^L x_{i,j}b_{i,j}\leq \frac{B_t(C_p (Q_{t-1}^P-\lambda)+T)}{s_{fov}(1-k_{dof}y_{dof})T}$.
The objective function includes two parts: 1) video quality loss $\Phi_t$ and 2) cybersickness accumulation $Q^{S}_{t}$ (i.e., sickness queue occupancy).
The video quality deteriorates if the viewport distortion is larger \cite{360vid-ProbDASH} and the viewport size (narrowed down by FoV shrinking) becomes smaller \cite{Sickness-FoV-2}. Also, DoF simulation undermines video quality since it blurs the out-of-focus parts of video \cite{Sickness-DoF,Sickness-DesignGuideline}. Accordingly, following \cite{360vid-ProbDASH}, the video quality loss is $\Phi_t = \frac{D_t}{s_{fov}(1-k_{dof}y_{dof})}$.
Therefore, the objective is $\xi\Phi_{t}+\rho Q^{S}_{t}$, where $\xi$ and $\rho$ are tuning knobs to trade-off video quality against cybersickness.\footnote{\color{c2}For the videos with highly dynamic contents (e.g., action movies), we can set a larger $\rho$ to focus more on cybersickness since they contain a large amount of optical flow and aggravate cybersickness more easily. For static videos like dramas, we can focus more on video quality and set a larger $\xi$. 
}
\ProbName\ is formulated as follows. 

\begin{definition}[\ProbName]
Given the duration $T$ of each video chunk cropped into $N$ tiles with $L$ quality levels in a frame, video metadata ($b_{i,j}$,$d_{i,j}$,$f_{i,j}$), packet queue occupancy $Q_{t}^P$, sickness queue occupancy $Q_{t}^S$, and bandwidth $B_t$, 
\ProbName\ finds 1) the fetched tiles and their quality, 2) the size of viewport (configured by FoV shrinking), and 3) the activation of DoF simulation, under the 1) \emph{fetching}, 2) \emph{packet queue length}, and 3) \emph{bandwidth} constraints.
The objective is to minimize the cost function of video quality loss and cybersickness accumulation over time, i.e., $\min \sum_{t=1}^{\infty} \xi\Phi_{t}+\rho Q^{S}_{t}$.
\end{definition}

%% file: 5-Algorithm.tex
\section{Online Algorithm} \label{alg}
To minimize the video quality loss, an intuitive approach is to classify the tiles into multiple groups with different priorities according to estimated viewing probability, and a group is assigned a higher quality if it has a higher priority \cite{360vid-BAS,360vid-Flare,360vid-TwoTier}. However, the above approach does not use FoV shrinking and DoF simulation to alleviate cybersickness.
We design an $m$-competitive online algorithm \AlgName\ to minimize the long-term video quality loss and cybersickness accumulation by examining Cybersickness Indicator (CI), Video Loss Indicator (VLI), and Sickness Migration Indicator (SMI) in every time slot.
CI and VLI respectively evaluate the instant cybersickness and the video quality loss.
Equipped with CI and VLI, \AlgName\ assigns the tile quality by evaluating the cost of the joint configuration of quality assignment, FoV shrinking, and DoF simulation for each tile in the viewport. 
When the packet queue overflows, it tends to choose high-quality tiles (requiring a longer transmission time) to slow down packet queue growth, and the faster cybersickness accumulation (due to a larger magnitude of optical flow) also needs to be alleviated by FoV shrinking and DoF simulation. 
{\color{c3}Moreover, \AlgName\ introduces SMI to jointly adjust tile quality and control the growth of sickness queue by evaluating the accumulated cybersickness and the increased optical flow.}
\AlgName\ includes three phases: 1) Viewport Prediction and Tile Selection (\AlgStepOne), 2) Tile Quality Assignment (\AlgStepTwo), and 3) Cybersickness and Tile Quality Control (\AlgStepThree).
The pseudocode of \AlgName\ is presented in \cite{tech} due to the limited space.

\subsection{Viewport Prediction and Tile Selection (\AlgStepOne)}
To adapt to user head rotation, \AlgStepOne\ derives the viewing probability $p_{i,t+1}$ for all tiles by the following steps: 1) \AlgStepOne\ first predicts the most likely viewport position in the next $T$ seconds. 
Specifically, let $(\theta^y_t,\theta^p_t)$ be the viewport position at time slot $t$, where $\theta^y_t$ and $\theta^p_t$ are the horizontal (yaw) and vertical (pitch) angles of the viewport, respectively. 
\AlgStepOne\ evaluates the viewport position at time slot $t+1$ by applying $(\omega^y_t T)$\textdegree\ of yaw rotation and $(\omega^p_t T)$\textdegree\ of pitch rotation from the current viewport position. The predicted viewport position is $(\theta^y_{t}+\omega^y_t T,\theta^p_{t}+\omega^p_t T)$.
2) \AlgStepOne\ extracts adjacent viewports to ensure that the fetched tiles can fully cover the user's view during the head rotation, based on the Gaussian distributions of yaw-axis and pitch-axis with the standard deviations $\sigma_y$ and $\sigma_p$, respectively \cite{360vid-ProbDASH}.
3) For each tile $i$, \AlgStepOne\ sums the viewing probability $P(\theta^y_{t+1},\theta^p_{t+1})$ of each predicted viewport position $(\theta^y_{t+1},\theta^p_{t+1})$ that overlaps with tile $i$ to derive $p_{i,t+1}$. That is, $p_{i,t+1}=\sum_{(\theta^y_{t+1},\theta^p_{t+1})\in\Theta} P(\theta^y_{t+1},\theta^p_{t+1})$, where $\Theta$ is the set of all the predicted viewport positions that overlap with tile $i$. 
\AlgStepOne\ then constructs a set of fetched tiles $V=\{v\ |\ p_{v,t+1}>0\}$ and examines the boundary subsets of $V$ to identify the tiles unlikely to be seen. Let $G_V$ be the boundary subset of $V$ formed by a row (column) of tiles\footnote{Since the viewport, composed of multiple tiles, needs to be rectangular \cite{360vid-summary}, 
a row (column) of tiles in a set is defined as the tiles sharing the same horizontal (vertical) coordinate, such as tiles $1, 7, 13$ in Fig. \ref{fig:2}.} in $V$.
If $p_{g,t+1}<\epsilon, \forall g \in G_V$ holds, we remove $G_V$ from $V$, where $\epsilon$ is the probability threshold to discard tiles.\footnote{Following \cite{360vid-TwoTier,360vid-EPASS360}, a smaller $\epsilon$ allows the mobile proxy to transmit a larger viewport to avoid the interruption caused by head rotation.}
For example, in Fig. \ref{fig:2}, the boundary subset removed is $\{12,18\}$ and the fetched tiles are $V=\{10,11,16,17\}$. Due to the space constraint, the detailed calculation of this and the remaining examples are provided in \cite{tech}.


\begin{figure}[t]
\centering
\includegraphics[width=0.4\textwidth] {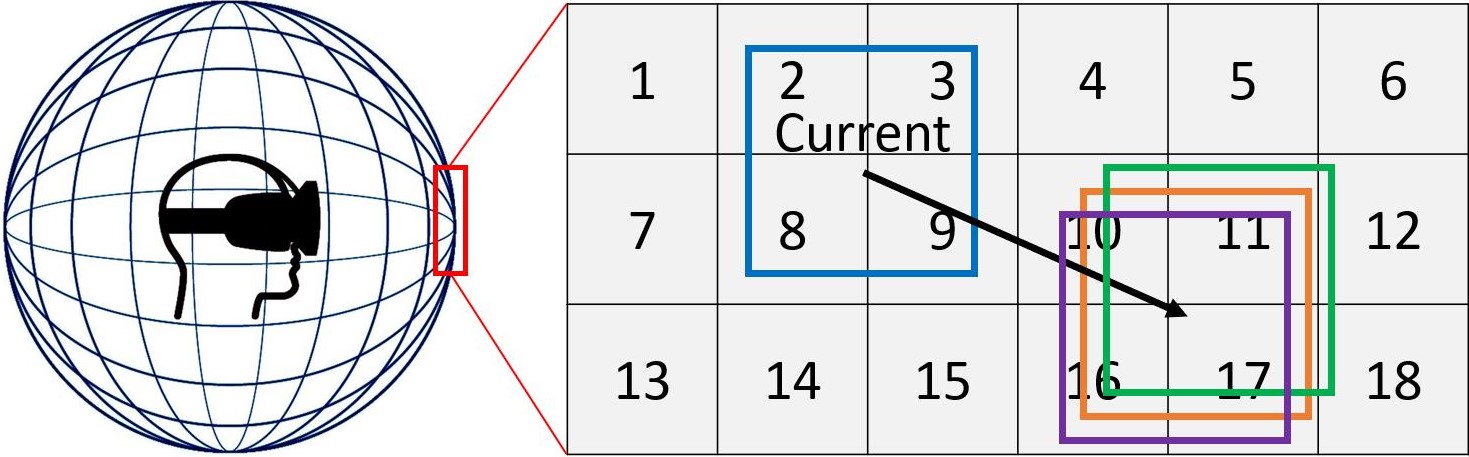}
\caption{An example of \AlgStepOne}
\label{fig:2}
\end{figure}

\subsection{Tile Quality Assignment (TQA)}
\AlgStepTwo\ assigns a quality level to each tile in $V$ based on the induced cost of the tile quality assignment 
and the bandwidth constraint. 
We introduce Cybersickness Indicator (CI) and Video Loss Indicator (VLI), weighted by the viewing probability $p_{i,t+1}$, to approximately evaluate the cybersickness and video quality loss.
{\color{c3}
A tile with more visual information has a larger size and more optical flow, which induces cybersickness. To alleviate cybersickness, FoV shrinking and DoF simulation reduce the size and resolution of the viewport such that the visual information (video size) decreases.
Therefore, we define $CI = \frac{f_{i,j}s_{fov}(1-k_{dof}y_{dof}) p_{i,t+1}}{\Sigma_{n\in V} (p_{n,t+1})}$ to measure the cybersickness induced by each tile regarding the effect of optical flow $f_{i,j}$, FoV shrinking $s_{fov}$, and DoF simulation $(1-k_{dof}y_{dof})$.}

{\color{c3}The video quality becomes worse when the distortion is larger and the viewport size and resolution are smaller (configured by FoV shrinking and DoF simulation). Hence, the video quality drops by the reciprocal of $s_{fov}(1-k_{dof}y_{dof})$.
We define $VLI=\frac{p_{i,t+1}d_{i,j}}{s_{fov}(1-k_{dof}y_{dof})\Sigma_{n\in V} (p_{n,t+1})}$ to measure the video quality loss induced by each tile regarding the effect of distortion $d_{i,j}$, FoV shrinking $s_{fov}$, and DoF simulation ($1-k_{dof}y_{dof}$).
}
Therefore, the cost of assigning quality $j$ to tile $i$ in $V$ with a certain configuration of FoV shrinking and DoF simulation $(s_{fov},y_{dof})$\footnote{We iteratively examine each configuration of $(s_{fov},y_{dof})$, and the number of iterations is small since $y_{dof}\in \{0,1\}$ and there are usually $7$ levels of $s_{fov}$ \cite{Sickness-FoV-1,Sickness-FoV-2}.} can be evaluated by $\tau(i,j,s_{fov},y_{dof})=\xi VLI(i,j,s_{fov},y_{dof}) + \rho CI(i,j,s_{fov},y_{dof})$. The function $\tau$ trades off the video quality against cybersickness.



The optimal quality of each tile may not be identical because the distortion and magnitude of optical flow of each tile vary. A tile with more visual information has a larger difference in the magnitude of optical flow between adjacent quality levels. Therefore, to jointly minimize the video quality loss and cybersickness accumulation, our idea is to iteratively derive the dynamic programming state of each quality assignment with the joint configuration of FoV shrinking and DoF simulation.
Let $M(n,\beta)$ be the minimum total cost of the quality assignment $A_{n,\beta}$ for the first $n$ tiles in $V$ under the bandwidth constraint $\beta$.
\AlgStepTwo\ systematically derives $M(n,\beta)$ according to the following two cases. 
1) When $n=0$ (i.e., no tile needs to be assigned quality), the total cost is $0$.
2) When $n>0$, \AlgStepTwo\ considers different quality assignments, which consume bandwidth no more than $\beta$, for tile $v_n \in V$. 
For each tile $v_n$, \AlgStepTwo\ recursively examines $M(n-1,\beta-b_{v_n,j})$, which is the minimum total cost of the quality assignment before assigning quality to $v_n$. Then, it chooses the minimum $M(n-1,\beta-b_{v_n,j})+\tau(v_n,j,s_{fov},y_{dof})$ to minimize the video quality loss and cybersickness accumulation.
Accordingly,
\begin{equation} \label{DP recursive function}
M(n,\beta)=
\begin{cases}
0, & \text{if } n=0 \\
\min\{M(n-1,\beta-b_{v_{n},j})+\tau\}, & \forall\ b_{v_{n},j}\leq \beta.
\end{cases}
\end{equation}

Recall that the total bandwidth for fetching tiles cannot exceed $\frac{B_t(C_p (Q_{t-1}^P-\lambda)+T)}{s_{fov}(1-k_{dof}y_{dof})T}$ to control the growth of the packet queue (detailed in Section \ref{pro}), and we let $\mathcal{B} = \frac{B_t(C_p (Q_{t-1}^P-\lambda)+T)}{s_{fov}(1-k_{dof}y_{dof})T}$. Hence, we can obtain the minimum total cost of the quality assignment for tile set $V$ by calculating {\color{c3}$M(n=|V|,\beta=\mathcal{B})$}.
Based on (\ref{DP recursive function}), \AlgStepTwo\ extracts a quality assignment to minimize the video quality loss and cybersickness accumulation. 
\AlgStepTwo\ first constructs a table
to store the bandwidth of each tile. 
Specifically, \AlgStepTwo\ stores the bandwidth $b_{v_n,j}$ of tile $v_n$ into the table with a label $(n,\beta)$ when it examines $M(n,\beta)=M(n-1,\beta-b_{v_n,j})+\tau(v_n,j,s_{fov},y_{dof})$ during the recursive process.
Afterward, \AlgStepTwo\ iteratively queries the table to construct $A_{n,\beta}$ since the bandwidth of each tile can be mapped to a certain quality and \AlgStepTwo\ can obtain the bandwidth of $v_{n-1}$ by querying the value labeled $(n-1,\beta-b_{v_n,j})$.


\textbf{Example 1.} Recall that we have obtained $V=\{10,11,16,17\}$ in \AlgStepOne. 
Assume that $B_t=8$, $C_p=4$, $Q_{t-1}^P=0.65$, $\lambda=0.5$, $T=1$, $k_{dof}=0.1$ and all possible configurations of $(s_{fov},y_{dof})$ are $(1,0)$, $(1,1)$, $(0.7,0)$, and $(0.7,1)$, which results in $\mathcal{B} = 13$, $14$, $18$, and $20$, respectively.
When $\mathcal{B}=13$, let $\{\tau,b_{i,1}\}=\{8,1\}$, $\{\tau,b_{i,2}\}=\{4,2\}$, $\{\tau,b_{i,3}\}=\{2,3\}$, and $\{\tau,b_{i,4}\}=\{1,4\},\forall i$. Then, $M(4,13)=\min\{M(3,9)+1,M(3,10)+2,M(3,11)+4,M(3,12)+8\}$. After the recurrence, the quality assigned to the tile set $\{10,11,16,17\}$ is $\{4,3,3,3\}$, since 
it induces the minimum total cost $1+2+2+2=7$. Similarly, for $\mathcal{B}=14$, $18$, and $20$, we have $\{4,4,3,3\}$, $\{4,4,4,4\}$, and $\{4,4,4,4\}$ as the quality levels of tile set $\{10,11,16,17\}$, respectively.

\subsection{Cybersickness and Tile Quality Control (CTQC)}
{\color{c3}To alleviate cybersickness accumulated over time, \AlgStepThree\ iteratively adjusts the tile quality assignment to slow down the growth of the sickness queue. When the cybersickness has accumulated massively, even slightly increasing optical flow can lead to the sickness queue overflow (i.e., the user will immediately quit watching 360\textdegree\ videos since the symptoms of cybersickness become unbearable \cite{Sickness-TempAspect,Sickness-UserAdaptation}). 
On the other hand, when the sickness queue occupancy is small, higher tile quality (which leads to more optical flow) can be used to reduce the video quality loss.
}
{\color{c3}Therefore, we introduce the Sickness Migration Indicator $SMI=\sum_{i\in V, j\in \Delta_V} (Q_{t-1}^S(f_{i,j}(t+1)-f_{i,j}(t))+VLI)$ integrated with VLI to jointly examine cybersickness and video quality, where $\Delta_V$ is the set of quality levels assigned to tile set $V$, $f_{i,j}(t+1)$ is the optical flow of tile $(i,j)$ at time slot $t+1$ (obtained from the video metadata). SMI evaluates the impact of accumulated cybersickness ($Q_{t-1}^S$) and future optical flow $f_{i,j}(t+1)$ on the tile quality assignment, and $f_{i,j}(t+1)-f_{i,j}(t)$ increases when $j$ becomes larger \cite{Sickness-OpticalFlow}.} 

To adjust the tile quality, we define $\Delta_V'$, a \emph{neighbor} of $\Delta_V$, as a quality assignment (satisfying the bandwidth constraint) that only one quality is different from $\Delta_V$ and the level difference is $1$. For example, $\Delta_V'$ is a neighbor of $\Delta_V=\{\ell_1,\ell_2\}$ if $\Delta_V'=\{\ell_1\pm 1,\ell_2\}$ or $\{\ell_1,\ell_2\pm 1\}$.
In the beginning, we let $\Delta_V$ be the quality assignment obtained in \AlgStepTwo\ and regard it as the \emph{center set} $\Delta_C$ for searching its neighbors to adjust tile quality. To avoid re-examining the same assignment in the next few iterations and improve the searching efficiency, \AlgStepThree\ maintains a neighbor search list (NSL) and inserts $\Delta_C$ into it.\footnote{The size of NSL can be adjusted based on the number of neighbors.}
\AlgStepThree\ first finds the neighbor $\Delta'_C$ (of the \emph{center set} $\Delta_C$), which induces the minimum SMI among the neighbors and is not in the NSL, as the \emph{center set} in the next iteration. Then, \AlgStepThree\ inserts $\Delta'_C$ into the NSL to avoid re-evaluating it in the next few iterations. If the list is full, \AlgStepThree\ will remove the oldest assignment. If $\Delta'_C$ induces a smaller SMI than $\Delta_V$, \AlgStepThree\ will replace $\Delta_V$ by $\Delta'_C$.
\AlgStepThree\ iteratively examines the above process until a certain quality assignment has been examined for $\alpha$ times, where $\alpha$ can be adjusted to prevent excessive computation time. 

\textbf{Example 2.} 
Following Example 1, when $\mathcal{B}=13$, $14$, $18$, and $20$, we have $\{4,3,3,3\}$, $\{4,4,3,3\}$, $\{4,4,4,4\}$, and $\{4,4,4,4\}$ as the quality levels of tile set $V=\{10,11,16,17\}$, respectively. Let $\Delta_V=\{4,3,3,3\}$ and its neighbor $\Delta_V'=\{3,3,3,3\}$. When $(s_{fov},y_{dof})=(1,0)$, the result of \AlgStepThree\ is $\Delta_V'$ since it has the minimum SMI. As for the other $(s_{fov},y_{dof})$ pairs, the quality assignments obtained from \AlgStepTwo\ remain the same in \AlgStepThree.

Last, \AlgName\ repeats \AlgStepTwo\ and \AlgStepThree\ with all possible FoV shrinking and DoF simulation configurations $(s_{fov},y_{dof})$ to minimize the video quality loss and cybersickness accumulation by selecting the tile quality assignment inducing the minimum $\xi\Phi_{t}+\rho Q^{S}_{t}$.
Following Example 2, after repeating \AlgStepTwo\ and \AlgStepThree\ for $(s_{fov},y_{dof})=(1,0)$, $(1,1)$, $(0.7,0)$, and $(0.7,1)$, 
we choose $\{4,4,4,4\}$ with $(s_{fov},y_{dof})=(0.7,0)$ since it induces the minimum $\xi\Phi_{t}+\rho Q^{S}_{t}$. 


\begin{theorem}
\label{theorem1}
    \AlgName\ is $\frac{1}{s_{min}(1-k_{dof})r}$-competitive in $O(NB+N^2)$ time, where $r$ is the minimum SSIM of all tiles, $s_{min}$ is the minimum value of $s_{fov}$, and $B=\max_{\forall t}\{B_t\}$.
\end{theorem}
\begin{proof}
The proof is presented in \cite{tech} due to the limited space.
\end{proof}
\vspace{-5mm}
\begin{theorem}
\label{theorem2}
\AlgName\ is $m$-competitive if $L\rightarrow\infty$.
\end{theorem}
\begin{proof}
The proof is presented in \cite{tech} due to the limited space.
\end{proof}

%% file: 6-Simulation.tex
\section{Simulations} \label{sim}
Due to the space constraint, some results and the user study are presented in \cite{tech}. 
\subsection{Simulation Setup}
We evaluate the performance of \AlgName\ with the video metadata extracted from 4K 360\textdegree\ videos \cite{dataset}.
The network bandwidth traces are extracted from the dataset HSDPA \cite{band}, and the mean value of the bandwidth is scaled from 3Mbps to 11Mbps \cite{360vid-EPASS360}. Following \cite{360vid-ProbDASH,360vid-EPASS360}, we cut the video into multiple chunks with the same duration of 1 second and split each chunk spatially into $6\times8$ tiles in default. Each tile includes 5 quality levels \cite{360vid-ProbDASH,360vid-Flare,360vid-EPASS360} and is encoded with H.264 using Constant Rate Factor (CRF) 43, 38, 33, 28, and 23 \cite{360vid-Flare} for level 1 to 5, respectively. We exploit the viewport prediction technique in \cite{360vid-ProbDASH} by formulating the prediction error with Gaussian distributions. The original viewport size is set to 100\textdegree$\times$100\textdegree in default \cite{360vid-EPASS360}. Following \cite{360vid-ProbDASH}, the capacity of packet queue is set to $C_p=4$s and the targeted queue occupancy is set to $\lambda=0.5$. $k_{dof}$ is set to 0.1 based on the result of applying DoF simulation to the videos \cite{Sickness-DoF,Sickness-DesignGuideline}. For the sickness queue, we follow \cite{Sickness-TempAspect} to set the capacity to $C_s=1000$ with the user adaptation to VR set to $\Omega=0.05$ in default. 
The default settings of $\xi$ and\ $\rho$ are 1 and 2.5, respectively. 
We compare \AlgName\ with three tile selection algorithms for 360\textdegree\ video streaming systems: 1) EPASS360 \cite{360vid-EPASS360}, 2) BAS \cite{360vid-BAS}, and 3) Flare \cite{360vid-Flare}. To evaluate \AlgName, we change the following parameters: 1) network bandwidth and 2) user adaptation to VR ($\Omega$), and we measure the performance metrics: 1) total cost, 2) video quality loss, 3) sickness queue occupancy, and 4) SSIM. Each result is averaged over 100 samples, and each sample includes 1,000+ video downloads.

\begin{figure}[t]
\centering
\subfloat[]{
    \centering
    \includegraphics[width=0.23\textwidth,height=2.5cm] {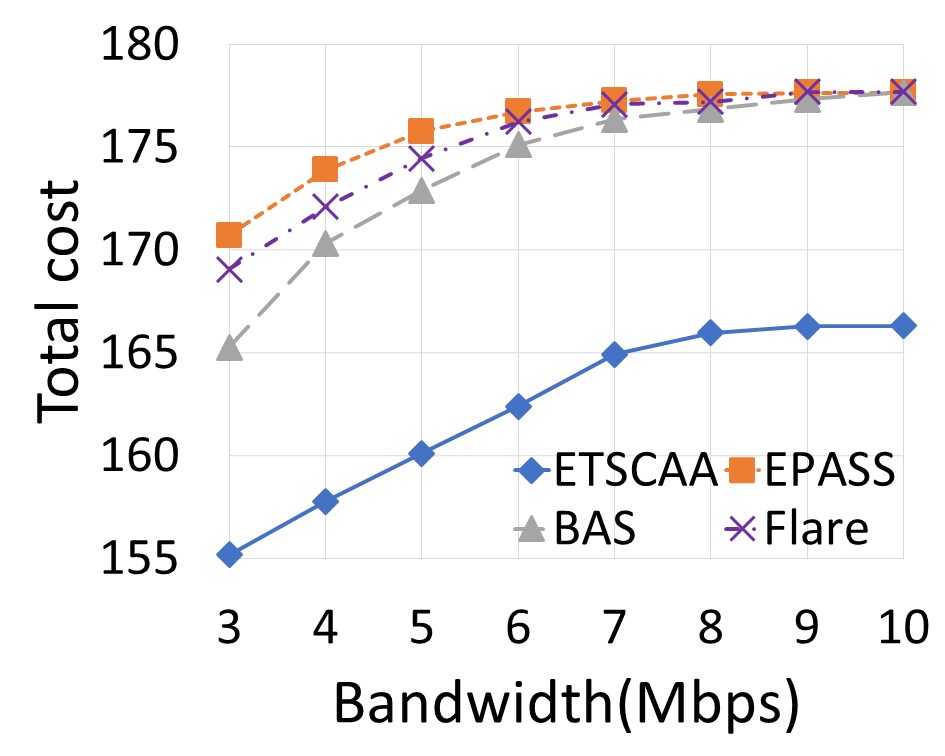}
    \label{fig:3-1}
}
\subfloat[]{
    \centering
    \includegraphics[width=0.23\textwidth,height=2.5cm] {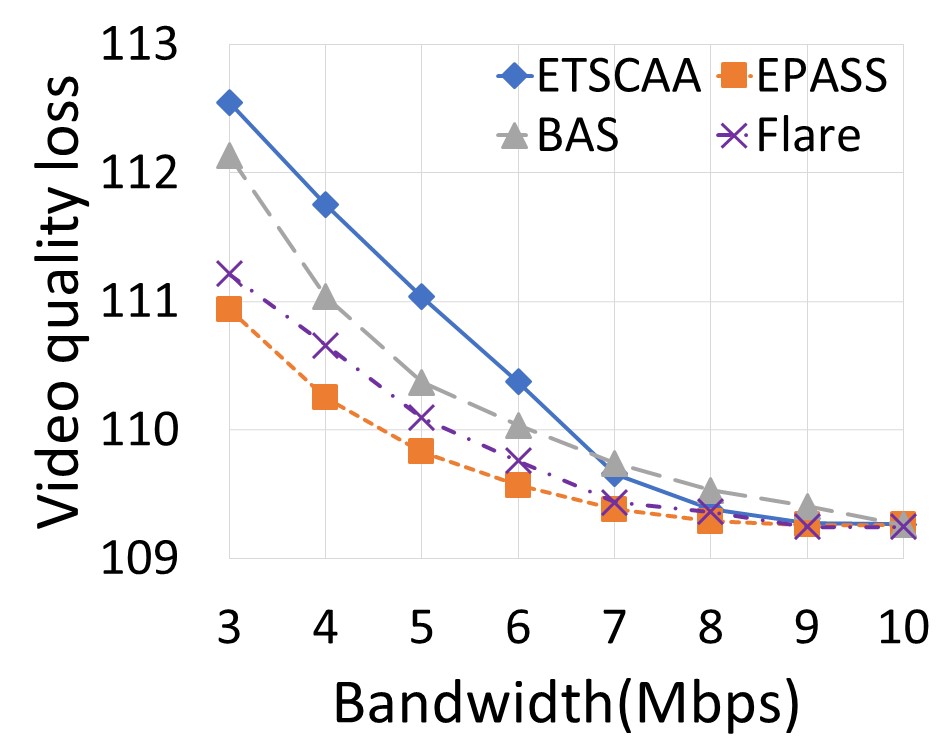}
    \label{fig:3-2}
}
\vspace{-3mm}
\subfloat[]{
    \centering
    \includegraphics[width=0.23\textwidth,height=2.5cm] {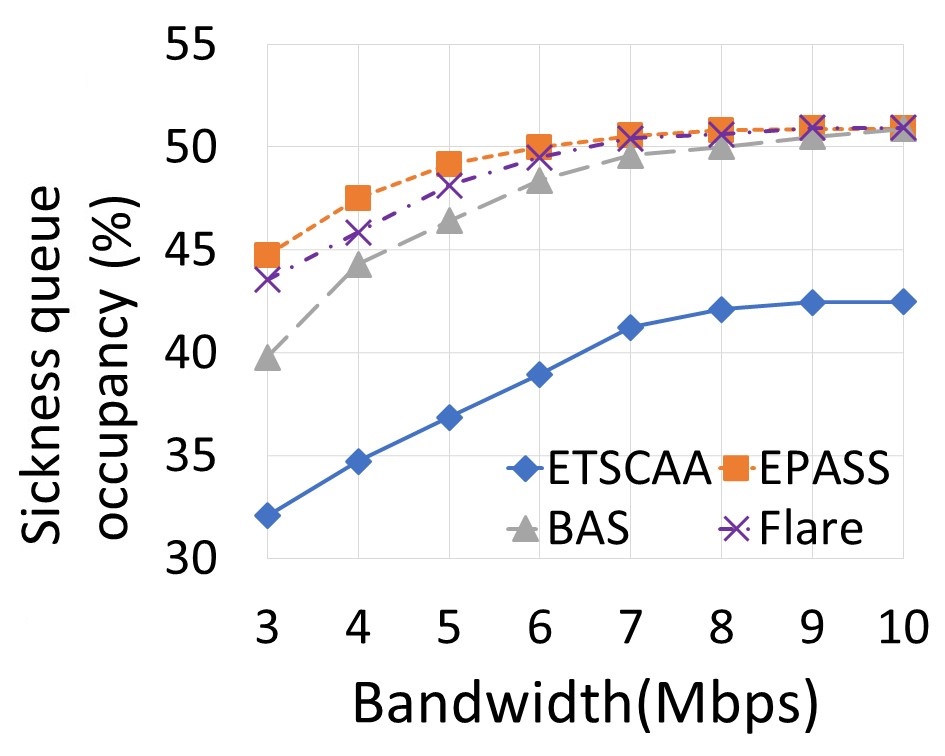}
    \label{fig:3-3}
}
\subfloat[]{
    \centering
    \includegraphics[width=0.23\textwidth,height=2.5cm] {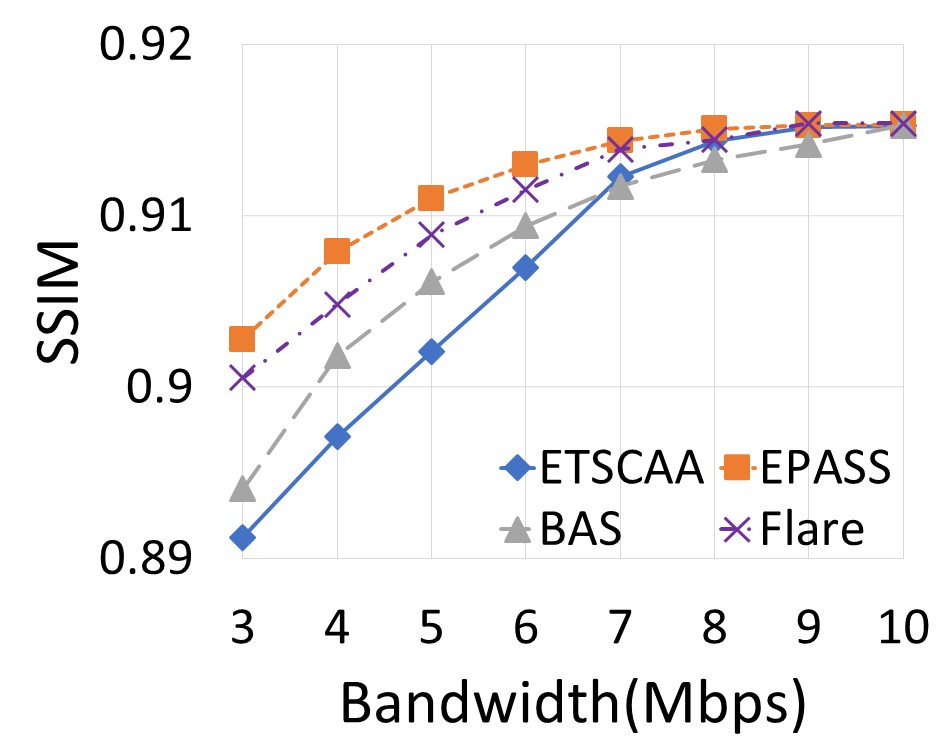}
    \label{fig:3-4}
}
\vspace{-3mm}
\subfloat[]{
    \centering
    \includegraphics[width=0.23\textwidth,height=2.5cm] {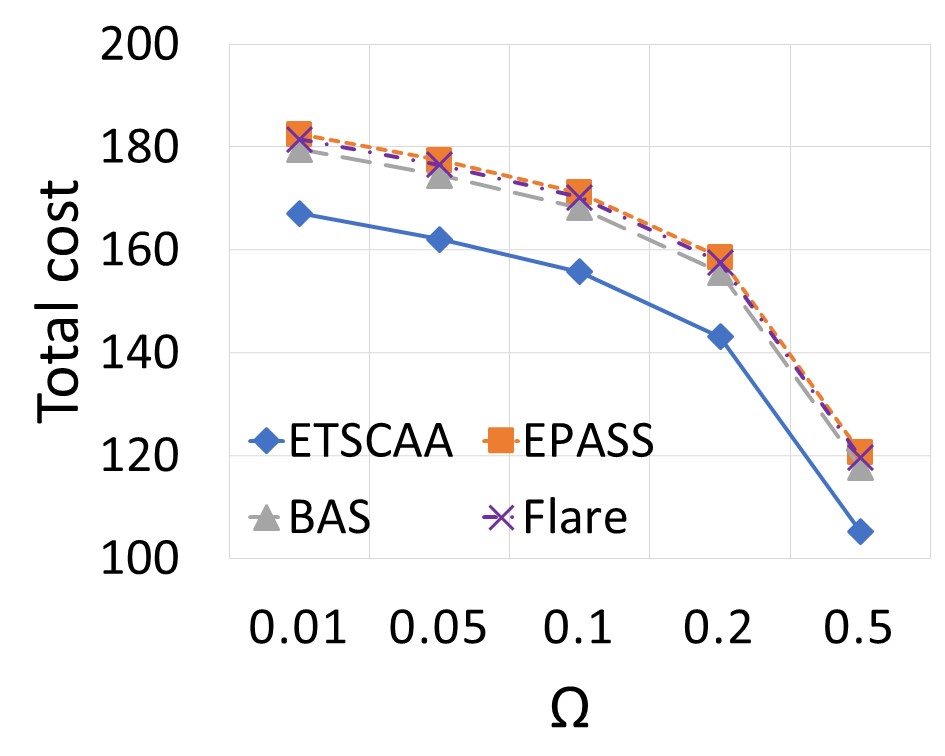}
    \label{fig:5-1}
}
\subfloat[]{
    \centering
    \includegraphics[width=0.23\textwidth,height=2.5cm] {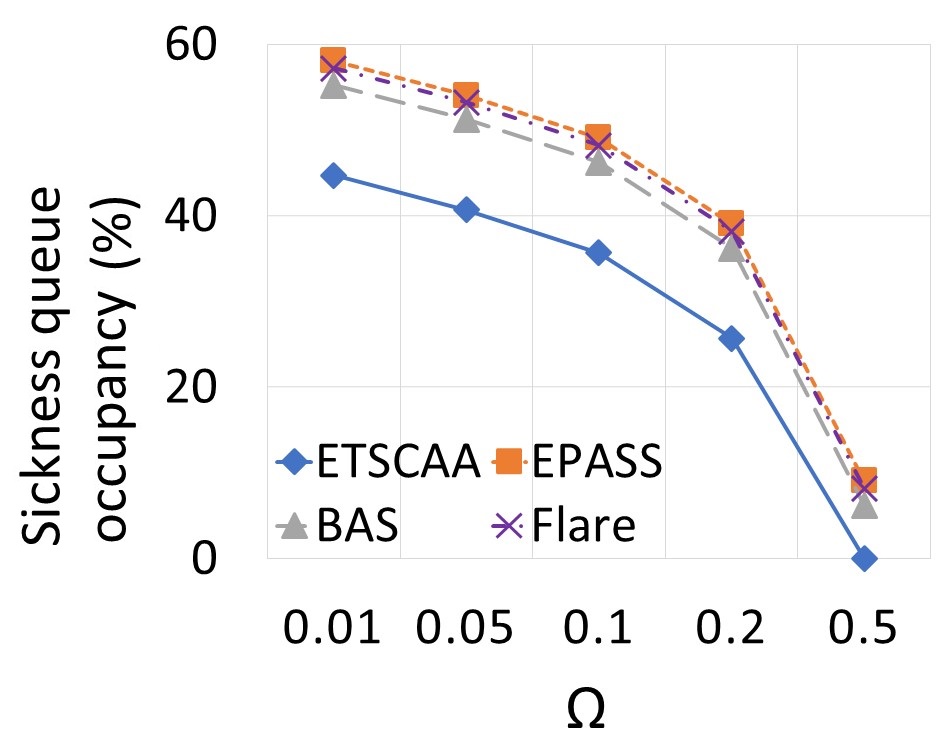}
    \label{fig:5-2}
}

\caption{Performance evaluation with different parameters.}
\label{fig:3}
\end{figure}

\subsection{Simulation Results}
In Fig. \ref{fig:3}(a), the total cost increases as the bandwidth grows since more tiles with higher quality (i.e., less distortion) are chosen, which leads to more accumulated cybersickness (see Fig. \ref{fig:3}(c)). 
{\color{c2}When the bandwidth is sufficient, \AlgName\ tends to choose high-quality tiles to reduce the video quality loss. 
However, it examines CI and VLI to activate FoV shrinking and DoF simulation to slightly mask and blur the outer part of the video for cybersickness alleviation. Therefore, \AlgName\ generates a much smaller sickness queue occupancy by degrading only a little video quality in Fig. \ref{fig:3}(c).}
{\color{c2}Comparing Figs. \ref{fig:3}(b) and \ref{fig:3}(c), the baselines lead to a smaller video quality loss but much higher sickness queue occupancy than {\AlgName}, because they maximize video quality by choosing the tiles with higher quality and viewing probability. 
In contrast, \AlgName\ exploits CI and VLI to carefully examine FoV shrinking and DoF simulation at different tiles, and they only sacrifice a little video quality (no more than $0.01$ SSIM \cite{SSIM-1}) in Fig. \ref{fig:3}(d) to achieve better cybersickness alleviation in Fig. \ref{fig:3}(c). 
}
Figs. \ref{fig:3}(e) and \ref{fig:3}(f) explore the impact of user adaptation to VR (i.e, $\Omega$).
The total cost and sickness queue occupancy plummet as $\Omega$ increases since users can adapt to VR more easily. 
{\color{c3}In contrast to the baselines maximizing the video quality, \AlgName\ evaluates the change of sickness queue occupancy over time by SMI to deal with the trade-off between cybersickness and video quality. 
Moreover, it exploits NSL to avoid re-evaluating recent tile quality configurations to ensure a better quality assignment in the next iteration. In general, \AlgName\ results in a smaller total cost and can reduce the cybersickness by more than $25\%$ compared to the baselines.}

%% file: 8-Conclusion.tex
\section{Conclusion} \label{con}
This paper leverages FoV shrinking and DoF simulation to alleviate cybersickness in a 360\textdegree\ video streaming system.
We formulate a new optimization problem \ProbName\ to minimize the video quality loss and cybersickness accumulation. To effectively solve \ProbName, we propose an $m$-competitive online algorithm \AlgName\ with CI, VLI, and SMI to support tile selection and cybersickness control. Simulations with a real network dataset demonstrate that \AlgName\ can reduce more than 25\% of sickness queue occupancy against the state-of-the-art tile selection algorithms for 360\textdegree\ video streaming systems.